\documentclass[letterpaper,10pt]{IEEEtran}

\usepackage{amsmath,amscd,amssymb} 

\usepackage{color,url,graphicx,epsfig,amsfonts,subfigure} 
\usepackage{cancel,enumerate,version,amsthm}
\newtheorem{theorem}{Theorem}[section]
\newtheorem{lemma}[theorem]{Lemma}
\newtheorem{proposition}[theorem]{Proposition}

\theoremstyle{definition}
\newtheorem{example}{Example}[section] 

\newtheorem{mydef}{Definition}[section]
\newtheorem{conj}{Conjecture}[section]

\theoremstyle{remark}
\newtheorem{remark}{Remark}

\newcommand{\ml}{\left[ \begin{array}{cccccccccccccccccccccccccc}}
\newcommand{\mr}{\end{array} \right]}
\newcommand{\mll}{\left[ \begin{array}{c|c|c|c|c}}

\newcommand{\yf}{\operatorname{fvct}}
\newcommand{\yn}{y^+}

\newcommand{\until}[1]{\{1,\dots, #1\}}

\newcommand\oprocendsymbol{\hbox{$\square$}}
\newcommand\oprocend{\relax\ifmmode\else\unskip\hfill\fi\oprocendsymbol}

\renewcommand{\theenumi}{(\roman{enumi})}

\newcommand{\real}{\mathbb{R}}
\renewcommand{\baselinestretch}{0.94}
\begin{document}
\title{On Opinion Dynamics in Heterogeneous Networks\thanks{This work was
    supported in part by the UCSB Institute for Collaborative Biotechnology
    through grant DAAD19-03-D004 from the U.S. Army Research Office.}}

\author{Anahita~Mirtabatabaei \hspace{0.5in} Francesco~Bullo\\
  \thanks{Anahita Mirtabatabaei and Francesco Bullo are with the Center for
    Control, Dynamical Systems, and Computation, University of California,
    Santa Barbara, Santa Barbara, CA 93106, USA,\tt{
      \{mirtabatabaei,bullo\}@engineering.ucsb.edu}}} \date{\today}
\maketitle

\begin{abstract}
  This paper studies the opinion dynamics model recently introduced by Hegselmann and Krause: each agent in a group maintains a real number describing its opinion; and each agent updates its opinion by averaging all other opinions that are within some given confidence range. The confidence ranges are distinct for each agent. This heterogeneity and state-dependent topology leads to poorly-understood complex dynamic behavior. We classify the agents via their interconnection topology and, accordingly, compute the equilibria of the system. We conjecture that any trajectory of this model eventually converges to a steady state under fixed topology. To establish this conjecture,  we derive two novel sufficient conditions: both conditions guarantee convergence and constant topology for infinite time, while one condition also guarantees monotonicity of the convergence. In the evolution under fixed topology for infinite time, we define leader groups that determine the followers' rate and direction of convergence. 
\end{abstract} 
\section{Introduction}
``Any social behavior can be viewed both as independent and dependent, as cause and effect'' \cite{FH:59}.
In a society, the impacts of individuals opinions on each other form a network, and as the time progresses, the opinions change as a function of such network's structure. Much research is done on how the topological properties of the interconnection network can effect final decisions. 
The study of opinion dynamics and social networks goes back to  J.R.P. French \cite{JF:56}.
French's Formal Theory of Social Power explores the patterns of interpersonal relations and agreements that can explain the influence process in groups of agents. Later F. Harary provided a necessary and sufficient condition to reach a consensus in French's model of power networks \cite{FH:59}.
 Besides sociology, opinion dynamics is also of interest in politics, as in voting prediction \cite{EBN:05}; physics, as in spinning particles \cite{EN-PK-FV-SR:03}; cultural studies, as in language change \cite{JT-MT-AN:07}; and economics, as in price change \cite{KS-RW:02}. 

An important step in modeling agents in economics has been switching from perfectly rational agents approach to a bounded rational, heterogeneous agents using rule of thumb strategy under bounded confidence. There is no trade in a world where all agents are perfectly rational, which is in contrast with the high daily trading volume. Having bounded confidence in a society, which accounts for homophily, means that an agent only interacts with those whose opinion is close to its own. Mathematical models of opinion dynamics under bounded confidence have been presented independently by: Hegselmann and Krause (HK model), where agents synchronously update their
opinions by averaging all opinions in their confidence bound \cite{RH-UK:02}; and by Deffuant and Weisbuch and others (DW model), where 
agents follow a pairwise-sequential updating procedure in averaging \cite{GW-GD-FA-JN:02}. 
HK and DW models can be classified into either agent-based and density-based models, if the number of agents is finite or
infinite, respectively; or into heterogeneous and homogeneous models, if the confidence bounds are uniform or agent-dependent, respectively.
The convergence of both agent- and density-based homogeneous HK models are
discussed in \cite{VDB-JMH-JNT:09}. 
 The agent-based homogeneous HK system
is proved to reach a fixed state in finite time \cite{JD:01}, the time
complexity of this convergence is discussed in \cite{SM-FB-JC-EF:05m}, its
stabilization theorem is given in \cite{JL:05a}, and its rate of
convergence to a global consensus is studied in \cite{AO-JNT:07}.  
The heterogeneous HK model is studied by Lorenz, who reformulated the HK
dynamics as an interactive Markov chain \cite{JL:06b} and analyzed the
effects of heterogeneous confidence bounds \cite{JL:09}. %
In this paper, we focus on discrete-time agent-based heterogeneous HK (htHK) model of opinion dynamics, whose dynamics is considerably more complex than the homogeneous case. The convergence of this model is
experimentally observed, but its proof is still an open problem.

As a first contribution, based on extensive numerical evidence, we conjecture that there exists a finite time along any htHK trajectory, after which the topology of the interconnection network remains unchanged, and hence the trajectory converges to a steady state. We partly prove our conjecture: (1) We design a classification of agents in the htHK system, which is a function of state-dependent interconnection topology;
 (2) We introduce the new notion of \emph{final value at constant topology}, characterize its properties, including required condition for this value to be an equilibrium vector, and formulate the map under which final value at constant topology is an image of current opinion vector; (3)  For each equilibrium opinion vector, we define its
\emph{equi-topology neighborhood} and \emph{invariant equi-topology
  neighborhood}. We show that if a trajectory enters the invariant
equi-topology neighborhood of an equilibrium, then it remains
confined to its equi-topology neighborhood and sustains an interconnection
topology equal to that of the equilibrium. This fact establishes a
novel and simple sufficient condition under which the trajectory converges to a steady state, and the topology of the interconnection
network remains unchanged.
(4) We define a rate of convergence as a function of final value at constant topology. Based on the direction of convergence and the defined rate, we derive a sufficient condition under which the trajectory \emph{monotonically} converges to a steady state, and the topology of the interconnection network remains unchanged. 
(5) We explore some interesting behavior of classes of agents 
when they update their opinions under fixed interconnection topology for infinite time, for instance, the existence of a leader group for each agent that determines the follower's rate and direction of convergence.

This paper is organized as follows. The mathematical model, conjecture, agents classification, and equilibria are discussed in Section~\ref{model}. The two sufficient conditions for convergence and the analysis of the evolution under fixed topology are presented 
in Section~\ref{sufficient_cond}. Conclusion and future work are given in Section~\ref{conclusion}. Finally, the Appendix contains some proofs. 


\section{Heterogeneous HK Model}\label{model}
Given the confidence bounds $r=\{r_1,\dots,r_n\}\in\real^n_{>0}$, we
associate to each opinion vector $x(t)= y \in\real^n$ the \emph{proximity digraph}
$G_r(y)$ with nodes $\until{n}$ and edge set defined as follows: the
set of out-neighbors of node $i$ is $\mathcal{N}_i (y) =
\{ j\in\until{n} : |y_i - y_j| \le r_i \}$.
The heterogeneous HK model of opinion dynamics updates $x(t)$ according to
\begin{equation} \label{system}
  x(t + 1)  =   A(x(t)) x(t),
\end{equation}
where, denoting the cardinality of $\mathcal{N}_i (y)$ by $|\mathcal{N}_i(y)|$, the $i,j$ entry of $A(x(t) = y)$ is defined by
\begin{align*}
  a_{ij} (y) & = 
  \begin{cases} 
   \frac{1}{\displaystyle |\mathcal{N}_i(y)|},  & \mbox{if } j \in \mathcal{N}_i(y),\\ 
    0, & \mbox{if } j \notin \mathcal{N}_i(y).
  \end{cases}
\end{align*}

\begin{conj}[Constant-topology in finite time]\label{conjtau}
It is conjectured that along every trajectory in an htHK system~\eqref{system}, there exists a finite time $\tau$ after which the state-dependent interconnection
  topology remains constant or, equivalently, $G_r(x(t))=G_r(x(\tau))$ for all $t \ge \tau$. 
\end{conj}

This conjecture is supported by the extensive numerical results presented in \cite[Section 5]{AM-FB:11f}.
Here, let us quote some relevant definitions from the graph theory, e.g. see \cite{FB-JC-SM:09}. In a digraph, if there exists a directed path from node $i$ to node $j$, then $i$ is a \emph{predecessor} of $j$, and $j$ is a \emph{successor} of $i$.
A node of a digraph is \emph{globally reachable} if it can be reached from any 
other node by traversing a directed path. A digraph is \emph{strongly connected}
if every node is globally reachable. A digraph is \emph{weakly connected} if replacing all of its directed edges with undirected edges produces a connected undirected graph. 
A maximal subgraph which is strongly or weakly connected forms a \emph{strongly connected component} (SCC) or a \emph{weakly connected component} (WCC), respectively. Every digraph $G$ can be decomposed into either its SCC's or its WCC's. Accordingly, the \emph{condensation digraph} of $G$, denoted $C (G)$, can be defined as follows: the nodes of $C (G)$ are the SCC's of $G$, and there exists a directed edge in $C (G)$ from node $S_1$ to node $S_2$ if and only if there exists a directed edge in $G$ from a node of $S_1$ to a node of $S_2$. A node with out-degree zero is named a
\emph{sink}. Knowing that the condensation digraphs are acyclic, each WCC in $C(G)$ is acyclic and thus has at least one sink.

\subsection{Agents Classification} \label{class}

We classify the agents in an htHK system~\eqref{system} based on their interaction topology at each time step. 
For any opinion vector $y \in \real^n$, the components of $G_r(y)$ can be classified into three classes. 
A \emph{closed-minded component} is a complete subgraph and an SCC of $G_r(y)$ that is a sink in $C(G_r(y))$. 
A \emph{moderate-minded component}  is a non-complete subgraph and an SCC of $G_r(y)$ that is a sink in $C(G_r(y))$. 
The rest of the SCC's in $G_r(y)$ are called \emph{open-minded SCC's}. 
We define \emph{open-minded subgraph} to be the remaining subgraph of $G_r(y)$ after removing its closed and moderate-minded components and their edges. A WCC of the open-minded subgraph is  called an \emph{open-minded WCC}, which is composed of one or more open-minded SCC's. The evolution and initial proximity digraph  of an htHK system are illustrated in Figures \ref{figevolution} and \ref{proximity}.

Since $C(G_r(y))$ is an acyclic digraph, in an appropriate ordering, its adjacency matrix is lower-triangular \cite{FB-JC-SM:09}. Consequently, the adjacency matrix of $G_r(y)$ is lower block triangular in such ordering.
Following the classification of SCC's in $G_r(y)$, we put $A(y)$ into \emph{canonical form} $\overline{A}(y)$, by an appropriate
 \emph{canonical permutation matrix} $P(y)$, 
\begin{equation*}
\overline{A}(y) = P(y) A(y) P^T(y) = \ml C(y) & 0 & 0\\0 & M(y) & 0\\  \Theta_C(y) & \Theta_M(y)&  \Theta(y) \mr.
\end{equation*}
The submatrices $C(y)$, $M(y)$, and $\Theta(y)$ are the adjacency matrices of the closed, moderate, and open-minded subgraphs respectively, and are block diagonal matrices. Each entry in $\Theta_C(y) $ or $ \Theta_M(y)$ represents an edge from an open-minded node to a closed or moderate-minded node, respectively. The adjacency matrix $A(y)$  is a non-negative row stochastic matrix, and its nonzero diagonal establishes its aperiodicity.

\begin{remark}
Previously, (Lorenz, 2006) classified the agents of the htHK systems into two classes named \emph{essential} and \emph{inessential}. An agent is essential if any of its successors is also a predecessor, and the rest of agents are inessential \cite{JL:06}. 
It is easy to see that closed and moderate-minded components are essential, and open-minded components are inessential. 
\end{remark}
\begin{figure}
 \centering
    \includegraphics[width=3in,keepaspectratio]{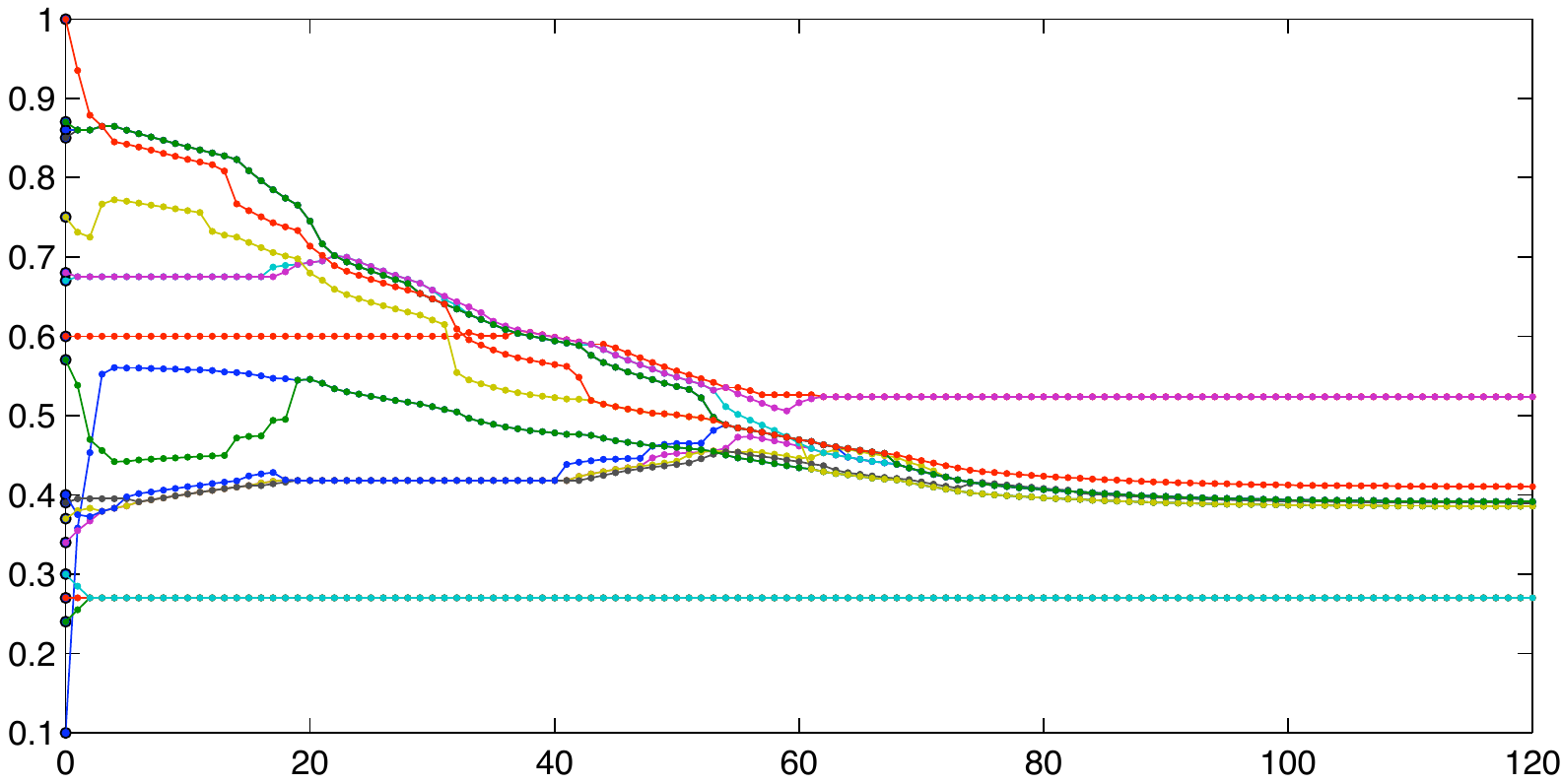}
 \caption{An htHK system evolution with initial opinion vector $x(0)=$ [0.1  0.24  0.27 0.3   0.34  0.37  0.39   0.4   0.5     0.6  0.67  0.68    0.75    0.85  0.86  0.87   1 ]$^T$, and bounds of confidence $r=$[0.5   0.04   0.04   0.04   0.031    0.021   0.011   0.061      0.25     0.01    0.04   0.03    0.3    0.07   0.07   0.07 0.135$]^T$. The interconnection topology of agents remains unchanged after $t=74$, see Conjecture~\ref{conjtau}.}\label{figevolution}
\end{figure} 
\begin{figure}
  \centering
  \includegraphics[width=3.4in,keepaspectratio]{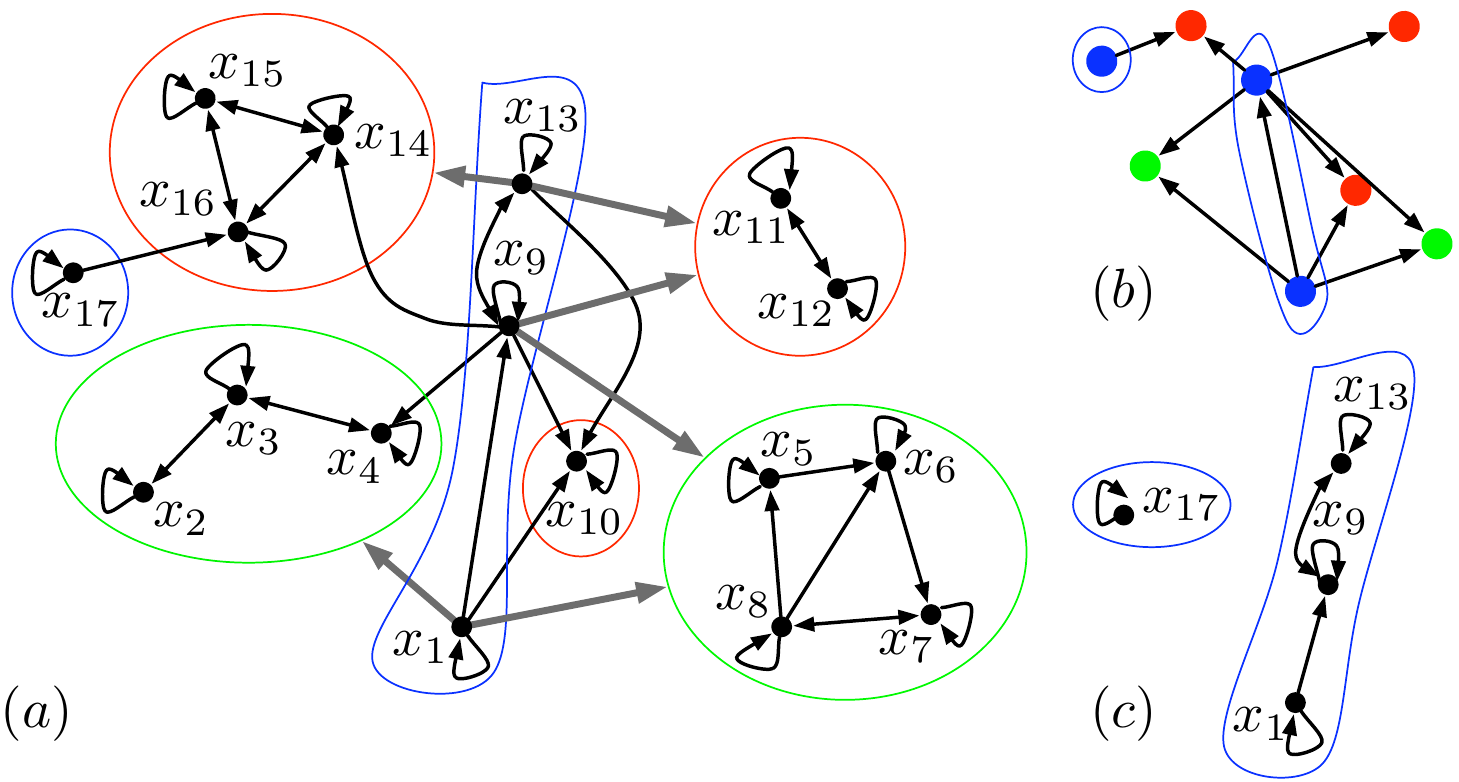}
  \caption{: (a) shows the proximity digraph $G_r(x(0))$ with its closed (red), moderate (green), and open-minded components (blue), and each thick gray edge represents multiple edges to all agents in one component; (b) is the condensation digraph $C(G_r(x(0)))$; and (c) is the open-minded subgraph.}\label{proximity} 
\end{figure}

\subsection{Final Value at Constant Topology}\label{final_value}
Based on Conjecture~\ref{conjtau}, for any opinion vector $y\in\real^n$ we define its \emph{final value at constant topology} $  \yf : \real^n \rightarrow \real^n$ to be $\yf(y) = \lim_{t\to\infty}A(y)^ty.$
Hence, if the interconnection topology of the system with initial opinion vector $y$ remains unchanged for infinite time, the opinion vector converges to $\yf(y)$.
Next, an opinion vector $y_0\in \real^n$ is an \emph{equilibrium} of the htHK system
if and only if $y_0$ is an eigenvector of the adjacency matrix $A(y_0)$ for eigenvalue one or, equivalently, $y_0 = A(y_0) y_0.$
The set of final values at constant topology is a superset of the equilibria. Clearly, if $y_0$ is an equilibrium, then its final value at constant
topology is equal to itself, that is, $\yf(y_0)=y_0$. The condition under which a final value at constant topology is an equilibrium of the system is discussed in Proposition~\ref{prop:properties-final-value}, for a proof of which see \cite{AM-FB:11f}.

\begin{proposition}[Properties of the final value at constant topology]
  \label{prop:properties-final-value}
  Consider opinion vector $y\in\real^n$:
  \begin{enumerate}
  \item $\yf(y)$ is well defined, and is equal to
    \begin{multline*}
      \yf(y) = P^T(y)  \\
      \ml C & 0 & 0\\ 
      0 & M^*& 0 \\
      (I - \Theta)^{-1}\Theta_CC  & (I - \Theta)^{-1}\Theta_MM^* & 0\mr(y) P(y) y, 
    \end{multline*}
    where $P(y)$ is the canonical permutation matrix,
    and $M^*(y) = \lim_{t\rightarrow \infty} M(y)^t$ is well defined.
  \item If the two networks of agents with opinion vectors $y$ and $\yf(y)$ have the same interconnection topology, or equivalently, $G_r(y) =G_r(\yf(y))$, then \label{nomoderate}
\begin{enumerate}
\item\label{iiparta} $\yf(y)$ is an equilibrium vector, 
\item\label{iipartb} $G_r(y)$ contains no moderate-minded component, and
\item\label{midclosed} in any  WCC of $G_r(\yf(y))$, the maximum and the minimum opinions $\yf_i(y)$ belong to that WCC's closed-minded components.
\end{enumerate}
  \end{enumerate}
\end{proposition}

\section{Convergence of htHK Systems}\label{sufficient_cond}
In this section, we present two sufficient conditions for htHK trajectories that guarantees fixed interconnection topology for infinite time and consequently convergence to a steady state. The second sufficient condition is more complicated and often more restrictive than the first condition, since it also guarantees the \emph{monotonicity} of the convergence.
 Moreover, there are examples in which the first condition does not hold while the second one does, see Example~\ref{znotxstarfig}. 
We justify the second sufficient condition by studying the behavior of htHK systems under fixed topology in a long run. 

\subsection{Convergence and Constant Topology}
The first sufficient condition for convergence is based on agents confidence bounds. According to this condition, if an htHK trajectory enters a specific neighborhood of any equilibrium of the system, then it stays in some larger neighborhood of that equilibrium for all future iterations, and its topology remains constant. Hence, the former neighborhood is a subset of the basin of attraction for the final value at constant topology of the entering opinion vector.
  \begin{mydef}[Equi-topology neighborhoods]
  Consider an htHK system with opinion vector $z\in \real^n$.
  \begin{enumerate} 
  \item  Define the vector $\epsilon(z)\in\real^n_{\geq0}$ with entries set equal to
$$\epsilon_i(z) = 0.5 \min_{j\in\until{n}\setminus\{i\}} \{||z_i - z_j| - R| : R \in \{r_i,r_j\} \}.$$
The \emph{equi-topology neighborhood} of $z$ is a set 
   of opinion vectors $y\in\real^n$ such that for all $ i\in\until{n}$,
  \begin{align*}
    |y_i - z_i| <  \epsilon_i(z), &  \text{ if } \epsilon_i(z) > 0, \text{ and}\\ 
    |y_i - z_i| =  \epsilon_i(z), &  \text{ if }   \epsilon_i(z) = 0.
  \end{align*}
\item Define the vector $\delta(z)\in\real^n_{\geq0}$ with entries set equal to
$$\delta_i(z) =\min \{\epsilon_j(z) : j \text{ is } i \text{'s predecessor in } G_r(z) \}.$$
   The \emph{invariant equi-topology neighborhood} of $z$ is a set 
   of opinion vectors $y\in\real^n$ such
  that for  all $ i\in\until{n}$,
  \begin{align*}
    |y_i - z_i| <  \delta_i(z), &  \text{ if }
    \delta_i(z) > 0, \text{ and}\\ 
    |y_i - z_i| =  \delta_i(z), & \text{ if }
    \delta_i(z) = 0.
  \end{align*}
  \end{enumerate}
\end{mydef}
\begin{theorem}[Sufficient condition for constant topology and convergence]\label{epsilonthm}
  Consider an htHK system with trajectory $x : \real \rightarrow \real^n$. Assume that there exists an equilibrium opinion vector $z\in\real^n$ such that $x(0)$ belongs to the invariant equi-topology neighborhood of $z$. Then, for all $t \ge 0$:
\begin{enumerate}
\item\label{epsst} $x(t)$ takes value in the equi-topology neighborhood of $z$,
\item $G_r(z) = G_r(x(t))$,
\item\label{nomoderatest} $G_r(x(t))$ contains no moderate-minded component, and
\item\label{convergencest} $x(t)$ converges to $\yf(x(0))$ as time goes to infinity. 
\end{enumerate}
\end{theorem}
This theorem is discussed and proved in \cite[Theorem 4.4]{AM-FB:11f}.

\subsection{Monotonic Convergence and Constant Topology}
The second sufficient condition for convergence is based on the rate and direction of convergence of the htHK trajectory in one time step. 
If a trajectory satisfies this condition, then any two opinions will either monotonically converge to each other or diverge from each other for all future iterations.

\begin{mydef}[Agent's per-step convergence factor]\label{k}
For an htHK trajectory $x(t)$, we define the \emph{per-step convergence factor} of an agent $i$ for which $x_i(t) - \yf_i(x(t)) \neq 0$  to be 
\begin{equation*}
k_i(x(t)) = \frac{x_i(t+1) - \yf_i(x(t)) }{ x_i(t) - \yf_i(x(t))}. 
\end{equation*}
\end{mydef}
The per-step convergence factor of a network of agents with distributed averaging was previously defined in~\cite{LX-SB:04} to measure the overall rate of convergence toward consensus. 
\begin{remark}[Monotonic convergence]\label{monotonicremark}
If the htHK trajectory $x(t)$ monotonically converges toward $\yf(x(t))$ in one time step, that is, for any $i \in \until{n}$,
\begin{equation*}
\begin{cases}
 x_i(t) \le   x_i(t+1) \le \yf_i(x(t)), \hspace{.1in}  & \mbox{ if }  x_i(t) <   \yf_i(x(t)), \\
x_i(t) \ge   x_i(t+1) \ge \yf_i(x(t)), \hspace{.1in} & \mbox{ if }  x_i(t) > \yf_i(x(t)), \\
x_i(t)  =   x_i(t+1) = \yf_i(x(t)), \hspace{.1in} & \mbox{ if }  x_i(t) = \yf_i(x(t)), 
\end{cases}
\end{equation*}
then
\begin{equation*}
\begin{cases}
0 \le k_i(x(t)) \le 1,  \hspace{.1in}  & \mbox{if }  k_i(x(t)) \mbox{ exists}, \\
x_i(t)  =  x_i(t+1) = \yf_i(x(t)), \hspace{.1in} &   \mbox{otherwise}.
\end{cases}
\end{equation*}
\end{remark}
Before proceeding, let us define the \emph{distance to final value} of any $y \in \real^n$ to be $\Delta(y) = y - \yf(y)$. 
For any open-minded agent $i$, let $k_{max_i}(y)$ and $k_{min_i}(y)$ denote the maximum and minimum per-step convergence factors over all $i$'s open-minded successors with nonzero distance to final value. Also, for any open-minded agents $i$ and $j$, let $k_{max_{i,j}}(y) = \max\{k_{max_{i}}(y), k_{max_{j}}(y)\}$ and $k_{min_{i,j}}(y) = \min\{k_{min_{i}}(y), k_{min_{j}}(y)\}$.
\begin{lemma}[Bound on per-step convergence factor]\label{convex} 
If in an htHK system with opinion vector $y\in\real^n$
\begin{enumerate}
\item $G_r(y)$ contains no moderate-minded component, and 
\item\label{deltabndforcovex} for any open-minded agent $i$ and any of its open-minded children $j$, $\Delta_i(y)\Delta_j(y)\ge 0$,
\end{enumerate}
then $k_i(A(y)y)$ is in the convex hull of $k_j(y)$'s.
\end{lemma}
\begin{theorem}[Sufficient condition for constant topology and monotonic convergence]\label{constanttopology}
  Assume that in an htHK system, the opinion vector $y\in\real^n$ satisfies
  the following properties:
  \begin{enumerate}
  \item\label{nbcond} the networks of agents with opinion vectors $y$ and $\yf(y)$  have the same interconnection topology, that is, $G_r(y) = G_r(\yf(y))$; 
  \item\label{ordercond} for any agents $i,j$, if $y_i\ge y_j$, then $\yf_i(y)\ge \yf_j(y)$; 
  \item\label{kcond}  $y$ monotonically converges to $\yf(y)$  in one iteration;
  \item\label{deltabndcond} for any open-minded neighbors $i$ and $j$, $\Delta_i(y)\Delta_j(y)\ge0$; 
  \item\label{kbndcond} any open-minded agents $i$ and $j$ that belong to the same WCC of $G_r(y)$ and that have nonzero $\Delta_i(y)$ and $\Delta_j(y)$,  have the following property: \\
a) if the sets of open-minded children of $i$ and $j$ are identical, then $k_i(y)=k_j(y)$, and\\
b) otherwise, assuming that $\Delta_i(y)\ge \Delta_j(y)$, 
\end{enumerate}
\begin{multline*}
  k_{max_{i,j}}(y) - k_{min_{i,j}}(y) \le   \min\{1- k_{max_{i,j}}(y), k_{min_{i,j}}(y) \} \\
\times  \min \{ \Big|1 - \frac{\alpha^m\Delta_j(y)}{\beta^m\Delta_i(y)}\Big| : \alpha \in [k_{min_j}(y),k_{max_j}(y)],
\\ \beta \in [k_{min_i}(y),k_{max_i}(y)], m\in \mathbb{Z}_{\geq0}\}
\end{multline*}
  Then the solution $x(t)$ from the initial condition $x(0)=y$ has the
  following properties: the proximity digraph $G_r(x(t))$ is equal to
  $G_r(y)$ for all time $t$, and the solution $x(t)$ monotonically
  converges to $\yf(y)$ as $t$ goes to infinity.
\end{theorem}
Lemma~\ref{convex} is employed in the proof of Theorem~\ref{constanttopology}, 
and the proofs to both are presented in the Appendix.
\begin{example}\label{znotxstarfig}
Consider an htHK trajectory with $x(0)=[0\;\; 0.6 \;\; 1]^T$ and confidence bounds $r=[0.5\;\; 1 \;\; 0.25]^T$, which converges with constant topology after $t=0$. It can be shown that $x(0)$ does not belong to any equilibrium vector's invariant equi-topology neighborhood. However, the sufficient condition of Theorem~\ref{constanttopology} eventually holds for this trajectory.
\end{example}
\subsection{Evolution under Constant Topology} 
Motivated by Conjecture~\ref{conjtau}, we investigate the rate and direction of convergence of an htHK trajectory whose interconnection topology remains constant for infinite time.
\begin{mydef}[Leader SCC]\label{leaderSCC}
Consider an htHK system with opinion vector $y\in\real^n$. 
For any open-minded SCC of $G_r(y)$, $S_k(y)$, denote the set of its open-minded successor SCC's by $\mathcal{M}(S_k(y))$, which includes $S_k(y)$.
We define $S_k(y)$'s \emph{leader SCC} to be an SCC in $\mathcal{M}(S_k(y))$ whose adjacency matrix has the largest spectral radius among all SCC's of $\mathcal{M}(S_k(y))$. 
\end{mydef}
\begin{theorem}[Evolution under constant topology]\label{convlemma}
Consider an htHK trajectory $x(t)$.
Assume that there exists a time $\tau$ after which $G_r(x(t))$ remains
  unchanged, that is, $G_r(x(t)) = G_r(x(\tau))$. Then, the following statements hold for all $t \ge \tau$: 
  \begin{enumerate}
  \item\label{convfirst} $\yf(x(t)) = \yf(x(\tau))$.
  \item\label{convsecond} $G_r(x(t))$ contains no moderate-minded component.
  \item\label{convthird} Consider any open-minded SCC of $G_r(x(t))$, $S_k(x(t))$, and its leader SCC $S_m(x(t))$, with adjacency matrices denoted by $\Theta_{k}$ and $\Theta_{m}$, respectively. Then,
\begin{enumerate}
\item\label{convthirda} for any $i \in S_k(x(t))$, either $x_i(t) - \yf_i(x(t)) = 0$ or its per-step convergence factor converges to the spectral radius of $\Theta_{m}$ as time goes to infinity, and
\item\label{convthirdb} if the spectral radius of $\Theta_{k}$ is strictly less than that of $\Theta_{m}$, then there exists $t_1 \ge \tau$ such that for all $i \in S_k(x(t))$, $j \in S_m(x(t))$, and $t \ge t_1$,
\begin{gather*}
      x_j(t_1) < \yf_j(x(t_1)) \quad\implies\quad  x_i(t) \leq \yf_i(x(t)),
      \\
      x_j(t_1) > \yf_j(x(t_1)) \quad\implies\quad  x_i(t) \geq \yf_i(x(t)).
    \end{gather*} 
 \end{enumerate}
\end{enumerate}
\end{theorem}
In above theorem, parts \ref{convthirda} and \ref{convthirdb} tell us, respectively, that the rates and directions of convergence of opinions in an open-minded SCC toward the final value at constant topology are governed by the rate and direction of convergence of its leader SCC.
In our htHK model, the adjacency matrix of a large SCC has a large spectral radius. Theorem~\ref{convlemma} demonstrates that the per-step convergence factor of such SCC is also large. Owing to the inverse relation between the per-step convergence factor of an agent and its rate of convergence toward the final value, the rate of convergence of a large open-minded SCC toward final opinion vector is small. Therefore, Theorem~\ref{convlemma} tells us that in a society with fixed interconnection topology, individuals converge to a final decision as slow as the slowest group of agents to whom they directly or indirectly listen. 
An example for the importance of convergence direction is that individuals follow their leaders in converging to a final price from low to high or vice versa. However, the final prices might be different, since they collect separate sets of information from closed-minded agents. A proof to Theorem~\ref{convlemma} and some numerical examples to facilitate the understanding of the conditions and results of the theorem are provided in~\cite{AM-FB:11f}.

\begin{remark}[Justification of the sufficient condition for monotonic convergence] 
We justify the conditions of  Theorem~\ref{constanttopology} employing Conjecture~\ref{conjtau} and Theorem~\ref{convlemma}. Note that these conditions are sufficient but not necessary for monotonic convergence.
Based on our conjecture, we assume that the topology of an htHK trajectory $x(t)$ remains unchanged after time $\tau$, thus condition~\ref{nbcond} of Theorem~\ref{constanttopology} is satisfied. Regarding conditions~\ref{ordercond} and \ref{kcond}, by statement~\ref{convthirda} of Theorem~\ref{convlemma}, there exist a time step $t_1 \ge \tau$, after which the per-step convergence factor of all agents belong to $[0,1]$. Therefore, the opinion vector converges toward its final value at constant topology monotonically in one step. Moreover, since the opinion vector is discrete, this monotonic convergence results in existence of a time step $t_2 \ge \tau$, after which condition~\ref{ordercond} of the Theorem~\ref{constanttopology} holds.
Regarding condition~\ref{deltabndcond}, statement~\ref{convthird} of Theorem~\ref{convlemma} shows that there exists time step $t_3 \ge \tau$, after which for any open-minded $i$ and $j$ it is true that: if they both belong to one SCC, then $\Delta_i(x(t)) \Delta_j(x(t)) \ge 0$; and if they belong to two separate SCC's with adjacency matrices $\Theta_{1}$ and $\Theta_{2}$, respectively, while $j$ is a successor of $i$, then when $\rho(\Theta_{1}) < \rho(\Theta_{2})$, often it is true that $\Delta_i(x(t)) \Delta_j(x(t)) \ge 0$, and when $\rho(\Theta_{1}) > \rho(\Theta_{2})$, $\Delta_j(x(t))$ converges to zero faster than $\Delta_i(x(t))$ and hence $\Delta_i(x(t)) \Delta_j(x(t)) \simeq 0$. 
Regarding condition~\ref{kbndcond} part (a), if $i$ and $j$ have the same set of open-minded children at time $t$, then $k_i(x(t+1)) = k_j(x(t+1))$, see proof of Theorem~\ref{constanttopology}.
Finally, we explain why the upper bound in condition~\ref{kbndcond} part (b) is less restrictive as time goes to infinity. Since, for such agent $i$, the distance to final values of all successors with smaller per-step convergence factors converge to zero, the interval $[k_{min_i}(x(t)),k_{max_i}(x(t))]$ reduces to one value, that is $k_{max_i}(x(t))$ to which $k_i(x(t))$ converges. Consequently, for large $t$, $ k_{max_{i,j}}(x(t))=\max\{k_i(x(t)),k_j(x(t))\}$ , $\alpha = k_j(x(t))$, and $\beta = k_i(x(t))$. Also, if $\Delta_i(x(t)) \ge \Delta_j(x(t))$, then $k_i(x(t)) \ge k_j(x(t))$, and hence 
$$\min_{m,\alpha,\beta}\Big|1 - \frac{\alpha^m\Delta_j(x(t))}{\beta^m\Delta_i(x(t))}\Big| \simeq 1-\frac{\Delta_j(x(t))}{\Delta_i(x(t))}.$$
A system may monotonically converge under fixed topology while condition~\ref{kbndcond} of Theorem~\ref{constanttopology} is not satisfied. However, Figure~\ref{example6}  illustrates the sufficiency of this condition.
\end{remark}
\begin{figure}
  \centering
$\begin{array}{ccc}
  \includegraphics[width=1.65in,keepaspectratio]{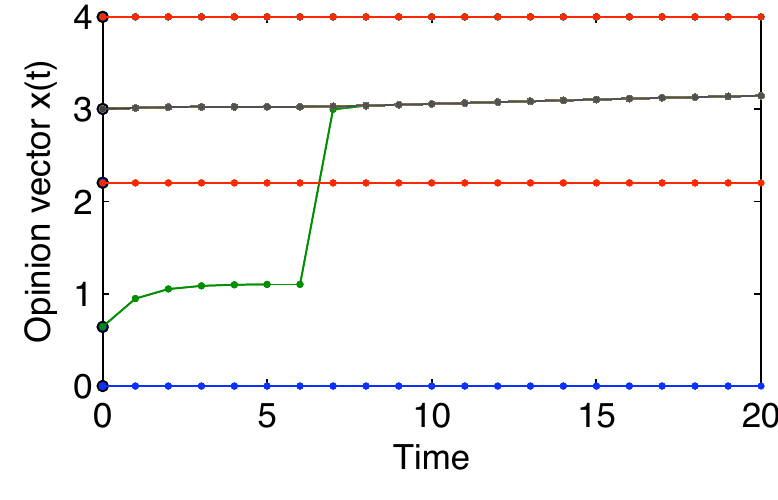} &
  \includegraphics[width=1.6in,keepaspectratio]{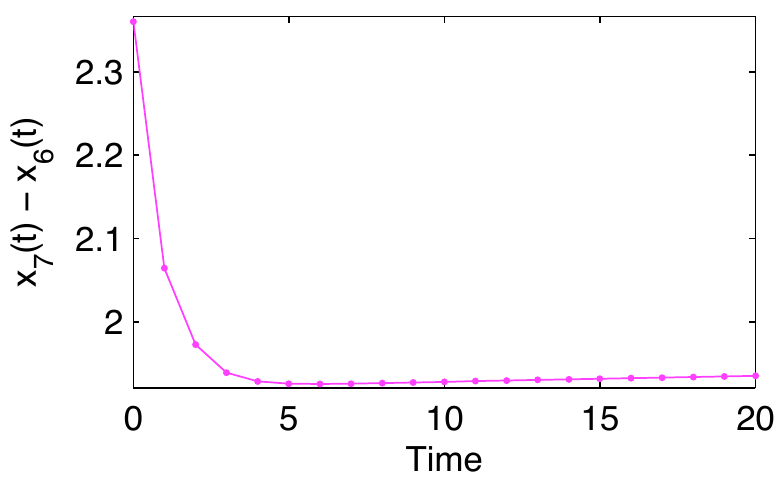}
\end{array}$
  \caption{Illustrates an htHK trajectory with $x_0  =  [0  \;\;    2.2  \;\;   4 \;\;    4\;\;     4\;\;    0.64 \;\;      3*{\bf 1}^T_{200}  ]^T$ and  $ r  =  [0.01  \;\;  0.01    \;\;  0.01    \;\;  0.01    \;\;   0.01  \;\;  1.9254   \;\;  2*{\bf 1}^T_{200} ]^T$ (left), and the non-monotonic evolution of the value $x_7(t) - x_6(t)$ if the proximity digraph remains fixed and equal to $G_r(x(0))$ (right), which is due to the large difference $k_6(x(t)) - k_7(x(t))$. The trajectory satisfies all conditions but \ref{kbndcond} of Theorem~\ref{constanttopology} at time steps $t=0,\dots,5$. 
The proximity digraph $G_r(x(0))$ contains two open-minded SCC's $\{x_6\}$ and $\{x_7,\dots,x_{206}\}$, who are two open-minded WCC's and weakly connected in $G_r(x(0))$. The per-step convergence factors of their agents, which is approximately equal to the spectral radius of the adjacency matrices of their SCC's (0.3333 and 0.9804), do not satisfy the boundary condition~\ref{kbndcond}. Therefore, the monotonic convergence of opinion vector, or equivalently equation~\eqref{distancebound}, does not hold.}
\label{example6}
\end{figure}

\section{Conclusion}\label{conclusion}
In this paper, we studied the heterogeneous HK (htHK) model of opinion dynamics. We provided two novel sufficient conditions that guarantee convergence and constant interconnection topology for infinite time, while one condition also guarantees monotonicity of convergence. 
Furthermore, we demonstrated that in the evolution under fixed topology, individuals converge to a final decision as slow as the slowest group to whom they directly or indirectly listen. 
The main future challenge is to prove the eventual convergence of all htHK systems. 
One approach is to verify that any trajectory is ultimately confined to the basin of attraction of an equilibrium. 

\bibliographystyle{plain}
\bibliography{../svn2/ref/alias,../svn2/ref/FB,../svn2/ref/Main}

\appendix
From here on, we often drop $y$ argument, and for any $y \in \real^n$, we denote $A(y) y$ by $\yn$ and $\yf(y)$ by $y^*$.  \renewcommand{\yf}{y^*}

\begin{proof}[Proof of Lemma~\ref{convex}]
If there is no moderate-minded component in $G_r(y)$, then $\yn_\Theta - \yf_\Theta = \Theta(y) (y_\Theta -\yf_\Theta)$, where $y_\Theta$ is the opinion vector of the open-minded class whose adjacency matrix is $\Theta(y)$, see \cite[Theorem 6.4]{AM-FB:11f}. 
Consider an open-minded agent $i$ whose children belong to the set $\{1,\dots,m\}$, and denote the entries of the adjacency matrix $A(y)$ by $a_{ij}$, then
\begin{multline}\label{kplus}
k_i(\yn) = \frac{a_{i1}(\yn_1 - \yf_1) + \dots + a_{im}(\yn_1 - \yf_m)}{a_{i1}(y_1 - \yf_1) + \dots + a_{im}(y_1 - \yf_m)}\\
  = \frac{a_{i1} k_1(y) \Delta_1(y)+ \dots + a_{im}k_m(y) \Delta_m(y)}{a_{i1}\Delta_1(y)+ \dots + a_{im}\Delta_m(y)}.
\end{multline}
Under condition~\ref{deltabndforcovex}, all $\Delta_j(y)$'s have the same sign, and hence all the terms in the right hand side are positive. 
Therefore, $k_i(\yn)$ is in the convex hull of $k_j(y)$'s.
\end{proof}
\begin{proof}[Proof of Theorem \ref{constanttopology}]
Here, we show that if $x(0) = y$ satisfies all the theorem's conditions, then $\yn$ also satisfies them, and similarly they hold for all subsequent times. Note that condition~\ref{kcond} guarantees entrywise monotonic convergence, and condition~\ref{nbcond} guarantees constant topology. 
Let us start by proving that $G_r(y) = G_r(\yn)$. On account of Proposition~\ref{prop:properties-final-value} part~\ref{nomoderate} and under condition~\ref{nbcond}, there are no moderate-minded component in $G_r(y)$, thus, for any $i,j\in\until{n}$, four cases are possible:

1. $i$ and $j$ are open-minded and weakly connected in $G_r(y)$. \\
a) If $\Delta_i \Delta_j > 0$, then without loss of generality we assume that $\Delta_i \ge \Delta_j > 0$, since otherwise we can multiply the opinion vector by $-1$. Hence, the monotonic convergence of the two opinions toward each other, or equivalently, 
\begin{equation}\label{distancebound}
\yf_i - \yf_j  \le \yn_i - \yn_j \le y_i - y_j,
\end{equation}
should be proved. Under condition~\ref{kbndcond}, it is true that
$|k_i - k_j| \le (1- \frac{ \Delta_j}{\Delta_i})\min\{ 1-k_j,  k_j \}.$
On the other hand,
\begin{multline*}
(\yn_i - \yn_j) - (\yf_i - \yf_j ) = (k_i - k_j) \Delta_i + k_j(\Delta_i -  \Delta_j) \\
\le  (1-k_j) \frac{\Delta_i- \Delta_j}{\Delta_i} \Delta_i + k_j(\Delta_i -  \Delta_j)  = \Delta_i -  \Delta_j,
\end{multline*}
which implies that $\yn_i - \yn_j  \le y_i - y_j.$ Furthermore,
\begin{multline*}
(\yn_i - \yn_j) - (\yf_i - \yf_j )  \ge  -|k_i - k_j| \Delta_i + k_j(\Delta_i -  \Delta_j)\\
\ge - k_j \frac{\Delta_i- \Delta_j}{\Delta_i} \Delta_i + k_j(\Delta_i -  \Delta_j)= 0,
\end{multline*} 
which implies that $\yn_i - \yn_j  \ge \yf_i - \yf_j.$ Now, we can show that the neighboring relation between $i$ and $j$ in the digraph $G_r(\yn)$ is equal to that of $G_r(y)$. We let $r$ denote either $r_i$ or $r_j$. The sign of $|y_i - y_j| - r$, $|\yn_i - \yn_j| -r$, and $|\yf_i - \yf_j| -r$ govern the neighboring relations between $i$ and $j$ in the digraphs $G_r(y)$, $G_r(\yn)$, and $G_r(\yf)$, respectively. Using inequalities~\eqref{distancebound} and condition~\ref{ordercond}
\begin{equation}\label{sangwich}
\begin{cases}
0 < \yf_i - \yf_j \le \yn_i - \yn_j \le y_i - y_j \hspace{.1in}  & \mbox{if} \hspace{.1in}  y_i \ge y_j, \\
\yf_j - \yf_i  \ge \yn_j - \yn_i \ge y_j - y_i > 0\hspace{.1in} & \mbox{if} \hspace{.1in} y_i \le y_j ,
\end{cases}
\end{equation}
subtracting $r$ from above inequalities gives
\begin{equation*}
\begin{cases}
|\yf_i - \yf_j| -r  \le  |\yn_i - \yn_j| -r  \le |y_i - y_j| -r   \mbox{ if }   y_i \ge y_j, \\
|\yf_j - \yf_i| -r  \ge |\yn_j - \yn_i| -r \ge |y_j - y_i| -r   \mbox{ if }  y_i \le y_j .
\end{cases}
\end{equation*}
Hence, $|\yn_i - \yn_j| -r$ is bounded between the two other values, which have the same sign by condition~\ref{nbcond}. Therefore, $i$ and $j$'s neighboring relation is preserved in $G_r(\yn)$. \\
b)  If $\Delta_i \Delta_j \le 0$, then for instance assume that  $\Delta_i \ge 0 \ge \Delta_j$. By condition~\ref{kcond}, it is easy to see that
$$ y_i - \yf_i \ge \yn_i - \yf_i \ge 0 \ge \yn_j - \yf_j  \ge y_j -\yf_j.$$
Using above inequalities and under condition~\ref{ordercond}, inequalities~\eqref{sangwich} hold, which again proves that $i$ and $j$'s neighboring relation is preserved in $G_r(\yn)$.

2. $i$ and $j$ are open-minded and belong to two separate WCC's of $G_r(y)$, whose agent sets are $\mathcal{V}_1$ and $\mathcal{V}_2$. Since $G_r(y) = G_r(\yf)$, by Proposition~\ref{prop:properties-final-value} part~\ref{midclosed}, the minimum and maximum opinions of a separate WCC in both $G_r(y)$ and $G_r(\yf)$ belong to closed-minded components. 
For any subgraph of $G_r(y)$, let us define its opinion range to be a closed interval in $\real$ between the minimum and maximum opinions of its agents; and its \emph{sensing range} to be the union of closed intervals in the confidence bounds of its agents around their opinions.
Therefore, the sensing range of $\mathcal{V}_1$ is separated from the opinion range of $\mathcal{V}_2$ in both $G_r(y)$ and $G_r(\yf)$. Due to monotonic convergence toward $\yf$ in one step, the sensing range of $\mathcal{V}_1$ in $G_r(\yn)$ lies in the union of its sensing ranges in $G_r(y)$ and $G_r(\yf)$. The boundary closed-minded component of $\mathcal{V}_1$ in $G_r(y)$ keeps the sensing range of $\mathcal{V}_1$ away from the opinion range of $\mathcal{V}_2$ in $G_r(\yn)$, see Figure~\ref{bothfig}~(a). Thus, two separate WCC's in $G_r(y)$ remain separate in $G_r(\yn)$.
\begin{figure}
  \centering
    \includegraphics[width=2.9in,keepaspectratio]{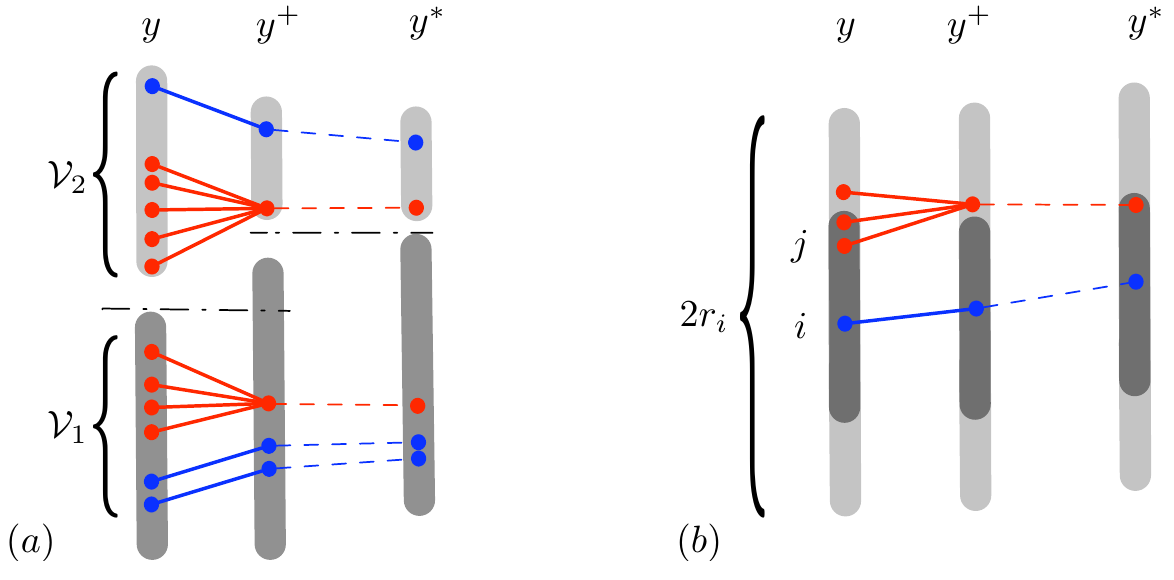}
  \caption{For the proof of Theorem~\ref{constanttopology}: (a) illusterates the sets of agents in two separate WCC's of $G_r(y)$, $\mathcal{V}_1$ and $\mathcal{V}_2 $. If $\mathcal{V}_1$'s sensing range (dark gray) is separated from $\mathcal{V}_2$'s opinion range (light gray) in $G_r(y)$ and $G_r(\yf)$, owing to boundary closed minded components (red), these ranges can not overlap in $G_r(\yn)$; and (b) shows that open-minded $i$ under light gray bound of confidence listens to closed-minded $j$ and its component in $G_r(y)$ and $G_r(\yf)$. Since $G_r(y) = G_r(\yf)$, closed-minded components reach consensus in $G_r(\yf)$. Otherwise, $i$ could listen to $j$ under dark gray bound of confidence, and get disconnected in $G_r(\yn)$.}    
\label{bothfig}
\end{figure} 

3.  $i$ and $j$ are both closed-minded in $G_r(y)$, hence, $\yn_i = \yf_i$ and $\yn_j = \yf_j$. The equality $ \yn_i - \yn_j = \yf_i - \yf_j $ tells us that neighboring relation between $i$ and $j$ in $G_r(\yn)$ is same as in $G_r(\yf)$, and consequently in $G_r(y)$. 

4. $i$ is open-minded and $j$ is closed-minded in $G_r(y)$. Since agents in one closed-minded component reach consensus in $G_r(\yf)$, $i$'s neighboring relation with $j$ in $G_r(y)$ is the same as its relation with other agents in $j$'s component. Assume that $y_i - y_j \le r_i$, see Figure~\ref{bothfig} (b), then $y_i - y_k \le r_i$ for all $k$ in $j$'s component. The average of the latter inequalities gives $y_i -\yn_j \le r_i $, and from $G_r(y) = G_r(\yf)$ we have $\yf_i - \yf_j \le r_i $, where for closed-minded $j$, $\yf_j = \yn_j$. Therefore, $\yn_i$, which under monotonic convergence is bounded between $y_i$ and $\yf_i$, also satisfies the inequality $\yn_i -\yn_j \le r_i $. Similarly, one can show that the neighboring relation is preserved in $G_r(\yn)$ for the case when $y_i - y_j > r_i$. 

So far, we have proved that $G_r(y) = G_r(\yn)$, hence condition~\ref{nbcond} holds for $\yn$. Due to monotonic convergence in one time step under opinion vector $y$, opinion order and direction of convergence toward final value is preserved in $\yn$, that is conditions~\ref{ordercond} and \ref{deltabndcond} are true for $\yn$.
To prove the last two conditions for $\yn$, we should find $k_i(\yn)$'s. 
Regarding part (a), if the two open-mindeds $i$ and $j$ have the same set of open-minded children, then equation~\eqref{kplus} tells us that $k_i(\yn) = k_j(\yn)$.
Regarding part (b), clearly, both conditions of Lemma~\ref{convex} hold for $G_r(y)$, hence for any open-minded $i$, $k_i(\yn)$ lies in the convex hull of $k_j(y)$'s, where $j$'s are its open minded children. This fact tells us that: $0 \le k_i(\yn) \le 1$, $k_{max_i}(\yn)\le k_{max_i}(y)$, and $k_{min_i}(\yn)\ge k_{min_i}(y)$.
Therefore, for any open-minded agents $i$ and $j$ with different sets of open-minded children,
\begin{multline*}
 k_{max_{i,j}}(\yn) - k_{min_{i,j}}(\yn) \le \min_{m,\alpha_1,\beta_1}\Big|1 - \frac{\alpha_1^m\Delta_j(y)}{\beta_1^m\Delta_i(y)}\Big|\\\times \min\{ 1- k_{max_{i,j}}(\yn), k_{min_{i,j}}(\yn) \} ,
\end{multline*}
where $\alpha_1 \in [k_{min_j}(y),k_{max_j}(y)]$, $\beta_1 \in [k_{min_i}(y),k_{max_i}(y)]$, and $m\in \mathbb{Z}_{\geq0}$. Knowing that $k_i(y) \in [k_{min_i}(y),k_{max_i}(y)]$,
\begin{gather*}
\min_{m,\alpha_1,\beta_1}\Big|1 - \frac{\alpha_1^m\Delta_j(y)}{\beta_1^m\Delta_i(y)}\Big| \le \min_{m,\alpha_1,\beta_1}\Big|1 - \frac{\alpha_1^m k_j(y)\Delta_j(y)}{\beta_1^mk_i(y)\Delta_i(y)}\Big|.
\end{gather*}
The right hand side of the above inequality is equal to
\begin{gather*}
\min_{m,\alpha_1,\beta_1}\Big|1 - \frac{\alpha_1^m \Delta_j(\yn)}{\beta_1^m\Delta_i(\yn)}\Big| \le \min_{m,\alpha_2,\beta_2}\Big|1 - \frac{\alpha_2^m \Delta_j(\yn)}{\beta_2^m\Delta_i(\yn)}\Big|,
\end{gather*}
where $\alpha_2$ and $\beta_2 $, respectively, belong to smaller intervals of $ [k_{min_j}(\yn) ,k_{max_j}(\yn)]$ and $[k_{min_i}(\yn),k_{max_i}(\yn)]$.
Hence, part (b) holds for $\yn$, which completes the proof.
\end{proof}

\end{document}